\documentclass[review,a4paper,sort&compress]{elsarticle}
\usepackage{hyperref}
\journal{Chaos, Solitons and Fractals}
\usepackage{amsfonts}
\usepackage{amsmath}
\usepackage{amsopn}
\usepackage{amssymb}
\usepackage{amsthm}
\usepackage{enumitem}
\usepackage{graphicx}

\usepackage{pifont}
\newcommand{\cmark}{\ding{51}}%
\newcommand{\xmark}{\ding{55}}%

\newcommand{\caselabel}[1]{\label{case:#1}}
\newcommand{\caseref}[1]{\ref{case:#1}}
\newcommand{\Case}[1]{Case~\caseref{#1}}
\newcommand{\case}[1]{case~\caseref{#1}}
\newcommand{\eqlabel}[1]{\label{eq:#1}}
\newcommand{\eq}[1]{\eqref{eq:#1}}
\newcommand{\figlabel}[1]{\label{fig:#1}}
\newcommand{\figref}[1]{\ref{fig:#1}}
\newcommand{\Fig}[1]{Figure~\figref{#1}}
\newcommand{\fig}[1]{Figure~\figref{#1}}
\newcommand{\proplabel}[1]{\label{prop:#1}}
\newcommand{\propref}[1]{\ref{prop:#1}}
\newcommand{\prop}[1]{Proposition~\propref{#1}}
\newcommand{\seclabel}[1]{\label{sec:#1}}
\newcommand{\secref}[1]{\ref{sec:#1}}
\newcommand{\secn}[1]{Section~\secref{#1}}

\newcommand{\tablabel}[1]{\label{tab:#1}}
\newcommand{\tabref}[1]{\ref{tab:#1}}
\newcommand{\Tab}[1]{Table~\tabref{#1}}

\newtheorem{proposition}{Proposition}

\providecommand{\abs}[1]{\left\lvert#1\right\rvert}             
\newcommand{\Complex}{\mathbb{C}}                               
\newcommand{\const}{\mathrm{const}}                             
\renewcommand{\d}{\mathrm{d}}                                   
\newcommand{\Df}[2]{\frac{\d{#1}}{\d{#2}}}                      
\newcommand{\Ddf}[2]{\frac{\d^2{#1}}{\d{#2}^2}}                 
\newcommand{\e}{\mathrm{e}}                                     
\newcommand{\iu}{\mathrm{i}}                                    
\newcommand{\Mx}[1]{\begin{bmatrix}#1\end{bmatrix}}             
\newcommand{\mx}[1]{\mathbf{#1}}                                
\renewcommand{\O}[1]{\mathcal{O}\left(#1\right)}                
\renewcommand{\P}[1]{\mathcal{P}_{#1}}                          
\renewcommand{\Re}[1]{\mathop{\mathrm{Re}}\left(#1\right)}      
\newcommand{\Real}{\mathbb{R}}                                  
\newcommand{\Zahlen}{\mathbb{Z}}                                

\newcommand{\+}[2]{\def#1{{#2}}}               
\newcommand{\1}[2]{\def#1##1{{#2}}}            
\+{\A}{A}                       
\+{\a}{a}                       
\+{\alp}{\alpha}                
\+{\ampf}{\nu}                  
\+{\ampfmax}{\ampf_{\max}}      
\+{\ana}{^*}                    
\+{\B}{B}                       
\+{\b}{b}                       
\+{\C}{C}                       
\+{\c}{c}                       
\+{\ch}{\chi}                   
\+{\D}{\mx{D}}                  
\+{\Dij}{D_{ij}}                
\+{\Du}{D_u}                    
\+{\Duu}{D_{uu}}                
\+{\Duv}{D_{uv}}                
\+{\Dv}{D_v}                    
\+{\Dvu}{D_{vu}}                
\+{\Dvv}{D_{vv}}                
\+{\eps}{\epsilon}              
\+{\f}{f}                       
\+{\fh}{\bar{f}}                
\+{\F}{\mx{F}}
\+{\g}{g}                       
\+{\h}{h}                       
\+{\j}{j}                       
\+{\k}{k}                       
\+{\lam}{\lambda}               
\+{\meps}{\varepsilon}          
\+{\n}{n}                       
\+{\num}{^{\sharp}}             
\+{\q}{q}                       
\+{\rh}{\rho}                   
\+{\sig}{\sigma}                
\+{\stept}{\Delta t}            
\+{\stepx}{\Delta x}            
\+{\t}{t}                       
\+{\Tinst}{T_{\mathrm{inst}}}   
\+{\Tbreak}{T_{\mathrm{break}}} 
\+{\u}{u}                       
\+{\uana}{\uh}                  
\+{\unum}{u\num}                
\+{\ueq}{u^*}                   
\+{\uh}{\hat{u}}                
\1{\ur}{u_{#1}}                 
\+{\veq}{v^*}                   
\+{\v}{v}                       
\+{\vana}{\vh}                  
\+{\vnum}{v\num}                
\+{\vh}{\hat{v}}                
\1{\vr}{v_{#1}}                 
\+{\vs}{v_\star}                
\+{\w}{\mx{w}}                  
\+{\weq}{\w^*}                  
\+{\wt}{\tilde{\w}}             
\+{\wavenumber}{\mu}            
\+{\x}{x}                       
\+{\xf}{\xi}                    
\+{\y}{y}                       

\begin{document}

\begin{frontmatter}
\title{Exact Propagating Wave Solutions in Reaction Cross-Diffusion System}
\author[alharj,exeter]{Abdullah Aldurayhim}
\author[exeter]{Vadim N. Biktashev\corref{mycorrespondingauthor}}
\cortext[mycorrespondingauthor]{Corresponding author}
\address[alharj]{
  Mathematics Department, %
  College of Science and Humanities in Al-Kharj, %
  Prince Sattam Bin Abdulaziz University, %
  Al-Kharj, Saudi Arabia %
}
\address[exeter]{
  College of Engineering, Mathematics and Physical Sciences, %
  University of Exeter, %
  Exeter EX4 4QF, UK %
}
\begin{abstract}
  Reaction-diffusion systems with cross-diffusion terms in addition to, or
  instead of, the usual self-diffusion demonstrate interesting features which
  motivate their further study. The present work is aimed at designing a toy
  reaction-cross-diffusion model with exact solutions in the form of
  propagating fronts. We propose a minimal model of this kind which involves
  two species linked by cross-diffusion, one of which governed by a linear
  equation and the other having a polynomial kinetic term. We classify the
  resulting exact propagating front solutions. Some of them have some features
  of the Fisher-KPP fronts and some features of the ZFK-Nagumo fronts.
\end{abstract}
\begin{keyword}
  Reaction-diffusion \sep
  cross-diffusion \sep
  Fisher-KPP model \sep
  ZFK-Nagumo model \sep
  propagating wave \sep
  propagating front
\end{keyword}
\end{frontmatter}


\section{Introduction}

Reaction-diffusion systems are models that are used widely to model physical,
chemical, biological and ecological processes.  Realistic models of such
processes are typically quite complicated and can only be dealt with
numerically.  However qualitative understanding of the most important features
benefits from analytical approaches, even if that requires sacrfices in
quantitative accuracy. This may be achieved by using asymptotic methods and/or
considering ``toy models''.

One of the first and famous ``toy'' reaction-diffusion
  systems is the model of propagation of an advantageous gene due to
  Fisher~\cite{Fisher-1937} and Kolmogorov, Petrovsky and
  Piskunov~\cite{Kolmogorov-etal-1937}. We refer to it as Fisher-KPP
model.
Another early archetypal reaction-diffusion equation was a model of flame
propagation considered by Zeldovich and
Frank-Kamenetsky~\cite{Zeldovich-FrankKamenetsky-1938}, which later became known also
as Schl\"ogl model~\cite{Schloegl-1972} and Nagumo
equation~\cite{McKean-1970}. We refer to it as ZFK-Nagumo equation.
Both models have monotonic propagating wavefront solutions of similar
appearances, but each has its own distinct mechanism.  The Fisher-KPP model
shows the transition from an unstable resting state to a stable resting state,
while the ZFK-Nagumo model shows the transition from one stable resting state
to another stable resting state.  Another qualitative difference between them
is that ZFK-Nagumo model exhibits a unique, up to a constant shift in time or
space, propagating front solution with a fixed speed and shape, whereas
Fisher-KPP model has a family of solutions with a continuous range of possible
speeds.
The importance of these toy models goes well beyond providing simplest
examples. For instance, the ZFK-Nagumo equation can be considered as
the fast subsystem in describing pulse waves in the FitzHugh-Nagumo
and similar systems using singular perturbation
techniques~\cite{Tyson-Keener-1988,Ikeda-etal-1989-PD}.

In the last decades, there
  has been ever increasing attention to reaction-diffusion systems
  which have cross-diffusion of the dynamic variables in addition or
  instead of their self-diffusion. These occur in mathematical
  modelling of various natural phenomena of biological, physical and
  chemical nature, such as mutual taxis of interacting species,
  including e.g.  %
  spatial segregation phenomena between the competing
  species~\cite{%
    Shigesada-etal-1979-JTB,%
    Iida-etal-2006-JMB,%
    Moussa-etal-2019-JNS%
  },
  cell types \cite{Sherratt-2000-PRSA} %
  and human population groups~\cite{Yizhaq-etal-2002-EPES}, and
  prey-taxis of predators and evasion of predators by prey~\cite{%
    Kerner-1959-BMB,%
    Kareiva-Odell-1987-AN,%
    Tsyganov-etal-2003-PRL,%
    Biktashev-Tsyganov-2005,%
    Wang-2006-CMA,%
    Tsyganov-etal-2007-UFN,%
    Lee-etal-2008-BMB,%
    Tsyganov-Biktashev-2014,%
    Biktashev-Tsyganov-2016-SR,%
    Negreanu-2020-DCDS%
  }; %
  interaction of populations of organisms or cells with environment,
  including e.g. %
  slime mold aggregation~\cite{Keller-Segel-1970-JTB}, %
  tumor angiogenesis~\cite{Sleeman-etal-2005-SIAM}, amoeboid
  locomotion~\cite{Ueda-etal-2011-PRE} %
  and %
  thermoregulation in honey bee
  colonies~\cite{Bastiaansen-etal-2020-SIAM}; %
  dissipative mechanical processes such as stick-slip motion of
  geological plates~\cite{Cartwright-etal-1997,
    Cartwright-etal-1999}; %
  as well as the literal cross-diffusion of reacting chemical
  species~\cite{Kirkaldy-1958-CJP,Vanag-Epstein-2009-PCCP,Gorban-etal-2011}. %
  Furthermore, cross-diffusion terms may appear ``mathematically'',
  via adiabatic elimination of fast but diffusing variables~\cite{%
    Kuramoto-1980-PTP,%
    Kuznetsov-etal-1994,%
    Iida-etal-2006-JMB,%
    Biktashev-Tsyganov-2016-SR,%
    Moussa-etal-2019-JNS%
  }. %
Interesting phenomena have been described in such systems,
  where the cross-diffusion plays an essential role alongside with the
  self-diffusion and  reaction part of the system.  This includes
  e.g. pattern formation via Turing-type instabilities~\cite{%
    Keller-Segel-1970-JTB,%
    Shigesada-etal-1979-JTB,%
    Kareiva-Odell-1987-AN,%
    Horstmann-2003-JDMV,%
    Iida-etal-2006-JMB,%
    Wang-2006-CMA,%
    Vanag-Epstein-2009-PCCP,%
    Lee-etal-2009-JBD,%
    Moussa-etal-2019-JNS%
  } and propagation of waves of various kinds~\cite{%
    Keller-Segel-1971-JTB-2,%
    Kuznetsov-etal-1994,%
    Sherratt-2000-PRSA,%
    Teramoto-etal-2004-PRE,%
    Lee-etal-2008-BMB,%
    Vanag-Epstein-2009-PCCP%
  }. %
  Overall, the literature on cross-diffusion models is too vast for an
  exhaustive survey here; some reviews of models and results with
  further references can be found e.g. in~\cite{%
    Shigesada-Kawasaki-1997-book,%
    Okubo-Levin-2000,%
    Horstmann-2003-JDMV,%
    Horstmann-Winkler-2005-JDE,%
    Vanag-Epstein-2009-PCCP,%
    Bellomo-etal-2015-MMMAS,%
    Painter-2019-JTB%
  }.
  
The focus of this work is on systems with excitable reaction
  kinetics, motivated by observations that including cross-diffusion
  in addition or instead of self-diffusion led to new
  phenomena~\cite{%
    Cartwright-etal-1997,%
    Cartwright-etal-1999,%
    Tsyganov-etal-2003-PRL,%
    Biktashev-Tsyganov-2005,%
    Tsyganov-etal-2007-UFN,%
    Tsyganov-Biktashev-2014%
  }. For example, propagating waves in reaction-cross-diffusion
  systems (RXD) with excitable reaction kinetics could penetrate each
  other on collision, a behaviour that is unusual for excitable
  systems with self-diffusion only.

The properties of solutions in RXD systems in the above
  cited motivating works have been mostly studied numerically. An
analytical approach has been attempted 
in~\cite{Biktashev-Tsyganov-2005}. In that work, fast-slow separation
between reaction kinetics of two reacting species is assumed. The fast
subsystem has piecewise linear kinetics and linear cross-diffusion,
and admits exact analytical solutions in the form of propagating
fronts. Unlike the Fisher-KPP and ZFK-Nagumo fronts, these front
solutions are oscillatory. They can be matched asymptotically with
slow pieces to obtain complete asymptotic description of propagating
pulses. The fast subsystem in this approach is different from the
Fisher-KPP and ZFK-Nagumo equations in two aspects: that it is
two-component and it is piecewise linear, as
opposed to the two ``classical''
toy models which are both one-component and with polynomial
nonlinearity of the kinetics. At least two components are of course
required to have cross-diffusion. 

In the present work, we 
investigate the possibility of having exact front solutions in a
cross-diffusion system with polynomial kinetics, unlike
  piecewise kinetics of \cite{Biktashev-Tsyganov-2005}. Our leading
idea is to postulate the solutions and deduce the governing equations
from there. For simplicity and as the first step, we only consider
here monotonic fronts, similar to those found in the ZFK-Nagumo
equation.  Thus it is clear for the outset
  that as far as are motivating numerical observations are concerned,
  the present study can only have a methodological value, as the waves
  observed in excitable cross-diffusion systems typically have
  oscillatory fronts and backs, as illustrated in~\fig{pulse_front}.

The paper is organized as follows.
The  problem formulation is given in \secn{main}.
In \secn{poly}, we consider the possibilities of chosing polynomial
nonlinearity for the reaction term.
In \secn{stability}, we discuss the simplest aspects of stability of
possible solutions.
Then we show the correspondent polynomial function suitable for solutions 
of the wavefront type.  These are presented in \secn{Correspondent}.
We demonstrate the possibility to have a wavefront
solution of the system as generalisation for Fisher-KPP in
  \secn{Possibilities_of_Generalising} and analyse the choices of the parameters needed to
imitate Fisher-KPP model in \secn{Imitate_Fisher}.
We return to the question of stability, now for the selected wavefront
solution, in \secn{Stability_of_zero_one}.
Results of numerical simulation are presented in \secn{Num_Simul}.  These
simulations show that the wavefronts are unstable.
These instabilities are investigated in
\secn{instability_FRONT_outer} and the paper is concluded by discussion in
\secn{discussion}.

\section{Problem formulation}
\seclabel{main}

Let us consider the reaction-diffusion system in the form
\begin{equation}
\begin{aligned}
\u_\t &= \f(\u) - \v + \Duv \v_{\x\x} + \Duu \u_{\x\x},
\\
\v_\t &= \eps (\u-\v) + \Dvu \u_{\x\x} + \Dvv \v_{\x\x},
\end{aligned}
\eqlabel{RXD_1}
\end{equation}
where 
\[
  \f(\u) = \u(\u-\alp)(1-\u), 
\]
and the parameters are restricted by $0\le \eps \ll 1$, $\alp \in (0,1/2)$.

The system \eq{RXD_1} is well studied as a
  reaction-self-diffusion system, with $\Duu>0$, $\Dvv\ge0$ and
  $\Duv=\Dvu=0$.  If $\Duv\ne 0$ and/or $\Dvu\ne 0$, we have
  reaction-cross-diffusion system.  %
  Regarding the signs of the diffusion coefficients, one common
  restriction is that their matrix must be positively semi-definite,
  so in particular, $\Duu\ge0$, $\Dvv\ge0$. Regarding the signs of the
  cross-diffusion coefficients, all sorts of combinations are
  considered in literature.  One of the ways the cross-diffusion terms
  as in \eq{RXD_1} may appear in applications is via linearization of
  terms describing mutual taxis of dynamic variables, which may
  describe populations and/or environmental factors affecting
  populations. For instance, if $\u$ represents a population which
  diffuses and moves towards attractant $\v$, which may be an
  environmental factor or a prey population and which itself only
  passively diffuses, then $\Duv<0$ and $\Dvu=0$, as e.g. in~\cite{%
    Keller-Segel-1970-JTB,%
    Keller-Segel-1971-JTB-2,%
    Kareiva-Odell-1987-AN,%
    Lee-etal-2008-BMB,%
    Lee-etal-2009-JBD%
  }.
  A similar combination (up to a change of sign of one of the dynamic
  variables) occurs in description of interaction of geological
  plates~\cite{Cartwright-etal-1997,Cartwright-etal-1999}.
  If $\u$
  and $\v$ represent competing species which seek to avoid each other,
  this leads to $\Duv>0$, $\Dvu>0$, as in~\cite{%
    Shigesada-etal-1979-JTB,%
    Iida-etal-2006-JMB%
  }.
  For predator-prey relationship, on the contrary, one may
  expect pursuit-evasion behaviour, that is, positive prey taxis for
  predators, i.e. predators seeking prey and prey escaping from
  predators, so if $\u$ component represents prey population and
  $\v$ represents predator population, this means that $\Duv>0$ and
  $\Dvu<0$, as in \cite{%
    Kerner-1959-BMB,%
    Kareiva-Odell-1987-AN,%
    Wang-2006-CMA,%
    Lee-etal-2008-BMB,%
    Negreanu-2020-DCDS%
  }. 
 Well-posedness of an initial or boundary-value
  problem for this system is not self-evident: examples are known that
  systems with cross-diffusion are capable of producing solutions
  blowing up in final time, see
  e.g. ~\cite{Horstmann-Winkler-2005-JDE}. Some well-posedness results
  have been established, see e.g.
  \cite{Meyries-etal-2014-SIAD,Negreanu-2020-DCDS}, however
  \cite{Meyries-etal-2014-SIAD} requires strong ellipticity of the
  diffusion matrix and \cite{Negreanu-2020-DCDS} requires strong
  stability properties of the reaction part of the system, neither of
  which is true in the case we consider. %
  We work under assumption that solutions exist and behave
  ``reasonably''; some evidence for that, even if not rigorous, is
  provided by the fact that the solutions can be simulated
  numerically. Clearly the well-posedness for the particular variants
  of the system of the form ~\eq{RXD_1} we consider here requires
  separate study. It is beyond the scope of this paper.  %

If $\eps=0$, $\v\equiv 0$, the self-diffusion system degenerates to
the ZFK-Nagumo equation
\cite{Zeldovich-FrankKamenetsky-1938,Schloegl-1972,McKean-1970} for
$\u(\x,\t)$, with an exact propagating front solution.  A piecewise
linear N-shaped variant of $\f(\u)$ also leads to exact propagating
front solution~\cite{McKean-1970}.  Qualitative properties of this
equation, including existence of propagating front solutions, persist
for a generic N-shape, and for $0<\eps\ll 1$, these solutions can form
a basis of asymptotic description, see for
instance~\cite{Tyson-Keener-1988,Ikeda-etal-1989-PD}.

A similar asymptotic approach for $0<\eps\ll 1$ was considered for the
cross-diffusion case of \eq{RXD_1} in~\cite{Biktashev-Tsyganov-2005}.  To
make the problem analytically tractable, the consideration there was
restricted to a piecewise linear N-shaped function $\f(\u)$ and pure
cross-diffusion, with self-diffion totally absent, $\Duu = \Dvv = 0$.

In this paper we consider the same system as was dealt with in in
\cite{Biktashev-Tsyganov-2005}, namely
\begin{equation}
\begin{aligned}
\u_\t &= \f \left( \u \right) - \v + \Dv \v_{\x\x}, \\
\v_\t &= \eps \left( \u-\v \right) - \Du \u_{\x\x},
\eqlabel{RXD_main}
\end{aligned}
\end{equation}
and intend to extend the methodology of \cite{Tyson-Keener-1988} and
\cite{Biktashev-Tsyganov-2005} for a polynomial function $\f(\u)$. In
absence of self-diffusion terms and in consideration of
  the chosen signs of the cross-diffusion coefficients, we abbreviate
  $\Dvu=\Du$ and $\Duv=-\Dv$.

We start by recapitulation of the approach
of~\cite{Biktashev-Tsyganov-2005} to set the scene and introduce
notation and terminology.  Direct numerical simulations of
\eq{RXD_main} with  cubic
$\f(\u)$ produces, in particular, solutions in the form of propagating
pulses of a fixed shape, as illustrated in~\fig{pulse_front}.  For
small $\eps$, the width of the pulse grows as $\O{\eps^{-1}}$.  This
means that in the limit $\eps \rightarrow 0$, the wave front and the
wave back of the pulse go apart. Our hypothesis
  is that for very small $\eps$, the system we are going to construct,
  will behave similarly to those discussed in \cite{Tyson-Keener-1988}
  and \cite{Biktashev-Tsyganov-2005}. Namely, we expect that a typical
  propagating wave solution will have the form of long stretches where
  $\u(\x,\t)$ remains near an instant equilibrium of the fast
  equation, satisfying $\f(\u)\approx\v$, which are interspersed by
  fast transitions from one such quasi-equibrium to another. Any such
  transition is approximated by an $\eps=0$ solution in the form of a
  wave which propagates with constant speed and shape and, far behind
  and far ahead, approaches constants, corresponding to the above
  mentioned quasi-equilibria. In particular, a pulse solution such as
  the one shown in \fig{pulse_front}, includes two such fast
  transitions, a front and a back.  Both the front and back represent
transitions between two distinct equilibrium points, say
$\left( \ur1 ,\vr1 \right)$ and $\left( \ur2 ,\vr2 \right)$.

\begin{figure}[tbhp]
  \centerline{\includegraphics{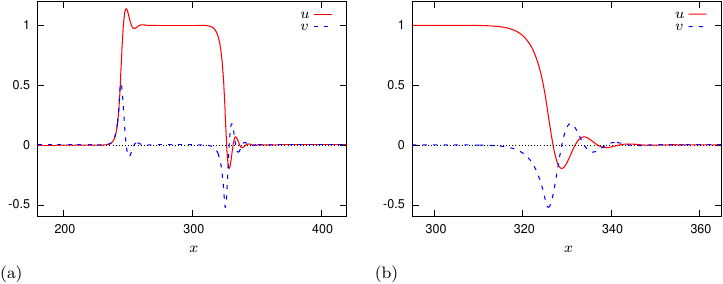}}
  \caption{%
    (a) The direct numerical simulation of \eq{RXD_main} the reaction
    cross-diffusion system with a no-flux boundary exhibits a propagating pulse
    with $\f(\u)=\u(\u-0.3)(1-\u)$ and the values of parameters are
    $\eps = 0.001$, $\Du = 5$ and $\Dv=0.5$. %
    (b) At small distance and time, the front of the pulse of reaction
    cross-diffusion system with a cubic nonlinearity approach to two asymptotic
    states $(\ur1,\vr1)$ and $(\ur2,\vr2)$. %
  }
  \figlabel{pulse_front}
\end{figure}

In the limit $\eps \rightarrow 0$ the system \eq{RXD_main} turns into
\begin{equation}
\begin{aligned}
  \u_\t &= \f \left( \u \right) - \v + \Dv \v_{\x\x}, \\
  \v_\t &= - \Du \u_{\x\x}.
\end{aligned}
\eqlabel{RXD_2}
\end{equation}
The two equilibria $\left( \ur1 , \vr1 \right)$ and
$\left( \ur2 , \vr2 \right)$, the asymptotic states of the wave front and the
wave back, satisfy $\f\left( \ur\j \right)=\vr\j$.  Let $\uh(\xf) = \u(\x,\t)$
and $\vh(\xf) = \v(\x,\t)$ be a propagating wave solution of \eq{RXD_2},
where $\xf = \x-\c\t$ and $\c >0$.  Substituting this into the system
\eq{RXD_2} yields
\begin{align}
  & \Dv \Ddf{\vh}{\xf} + \c \Df{\uh}{\xf} + \f \left( \uh \right) - \vh = 0, \eqlabel{RXD_5} \\
  & - \Du \Ddf{\uh}{\xf} + \c \Df{\vh}{\xf}  =0. \eqlabel{RXD_6}
\end{align}
As the front asymptotically approaches distinct steady states, we have
\begin{align}
  & \uh \left( \pm \infty \right) = \ur{1,2} , \quad
    \vh \left( \pm \infty \right) =  \vr{1,2}  \eqlabel{RXD_7} \\
  & \Df{\uh}{\xf} \left( \pm \infty \right) = 
    \Df{\vh}{\xf} \left( \pm \infty \right) = 0 .  \eqlabel{RXD_8}
\end{align}
Integrating \eq{RXD_6} with respect to $\xf$ gives
\begin{equation}
  \vh - \frac{\Du}{\c} \uh' = \vs = \const.               \eqlabel{RXD_9}
\end{equation}
When $\xf\to\pm\infty$, we obtain from \eq{RXD_9} that
$\vs = \vh_1 = \vh_2$ and then equation \eq{RXD_5} turns into
\begin{equation}\eqlabel{two-roots}
  \f\left( \ur{1,2} \right) = \vs.
\end{equation}
We have from equation \eq{RXD_7} that $\c\vh' = \Du \uh''$, hence
$\vh'' = \Du\uh'''/\c$.
Substituting this into \eq{RXD_6} yields
\begin{equation}
  \Dv \Du \uh'''
  + \left( \c^2 - \Du \right) \uh'
  + \c \left( \f (\uh) - \vs \right) = 0 , \qquad
  \uh(\pm \infty) = \ur{1,2}, \eqlabel{RXD_10}
\end{equation}
where $\uh$ is a wave solution for the reaction cross-diffusion system
\eq{RXD_2}.

This differential equation is deduced by applying the wave variable on the
reaction-cross-diffusion system \eq{RXD_2}.  Biktashev and Tsyganov
\cite{Biktashev-Tsyganov-2005} have replaced $\f(\uh)$ by a piecewise linear
function.  The fronts that are obtained from the piecewise linear function are
oscillatory fronts and are similar to those seen in
numerical simulations with cubic $\f(\uh)$.  We seek to consider a polynomial
function for $\f (\uh)$ instead of piecewise linear function, which would
still yield explicit analytical solutions for propagating fronts.

\section{Selecting the class of the polynomial reaction term}
\seclabel{poly} 

We aim to identify polynomial functions $\f(\uh)$ which would make the
differential equation \eq{RXD_10} analytically solvable.  First we write
the equation \eq{RXD_10} as
\begin{equation}
\A \uh''' + \B \uh' = \fh (\uh),
\eqlabel{RXD_11}
\end{equation}
where
\[
  \A = - \tfrac{\Du \Dv}{\c},   \quad
  \B = \tfrac{\Du - \c^2}{\c},  \quad
  \fh(\uh) = \f(\uh) - \vs. 
\]
We apply a reduction of order substitution,
\begin{align}
  \Df{\uh}{\xf} = \y \left( \uh \right). \eqlabel{RXD_12}
\end{align}
Substituting \eq{RXD_12} into \eq{RXD_11} gives
\begin{equation}
  \y \left[ \A \left( \y'^2 +\y \y'' \right) +\B  \right] = \fh(\uh) . \eqlabel{RXD_13}
\end{equation}
We aim that function $\fh(\uh)$ is a polynomial. This can be assured if
$\y(\uh)$ is a polynomial.

Let us find the possible degree of the polynomials $\y(\uh)$ and $\fh(\uh)$.
Let $\P\n$ be the set of polynomials of degree $\n$.  If $\y\in \P\n$, then
\[
  \y \left[ \A \left( \y'^2 +\y \y'' \right) + \B  \right] = \fh(\uh) \in \P{3\n-2} . 
\]
If $\n=1$ then $\fh(\uh)$ is linear, which is not of interest for us, as this
cannot produce two distinct solutions for~\eq{two-roots}.  If $\n=2$ then
$\fh(\uh)$ is quartic.  This quartic polynomial is comparable to cubic, in
that it can describe bistability, if it has at least three simple
roots. Therefore, $\y\in \P2$, $\fh\in \P4$ is the simplest suitable choice.
 
The travelling wave differential equation for ZFK-Nagumo can be solved
analytically by a reduction of order \cite{McKean-1970}. Incidentally, in
that solution $\y(\uh)$ is also quadratic.  It is convenient to factorise the
quadratic polynomial $\y (\uh)$,
\begin{equation}
  \y (\uh) = \k \left( \uh - \g \right) \left( \uh - \h \right),    \eqlabel{RXD_15}
\end{equation}
for some constants $\k\ne0$, $\g$ and $\h$.  Note that due to~\eq{RXD_7},
\eq{RXD_8} and \eq{RXD_12}, we have $\{\ur1,\ur2\}=\{\g,\h\}$.

From \eq{RXD_12} and \eq{RXD_15}, we obtain
\begin{equation}
  \uh(\xf) = \frac{\g+ \h\,\e^\ch}{1 + \e^\ch} , \qquad 
  \ch = \k (\xf + \C) (\g-\h) , \eqlabel{RXD_16}
\end{equation}
where $\C$ is an arbitrary constant. The front wave described by
\eq{RXD_16} is illustrated in \fig{tanh}.

\begin{figure}[tbhp]
  \centerline{\includegraphics{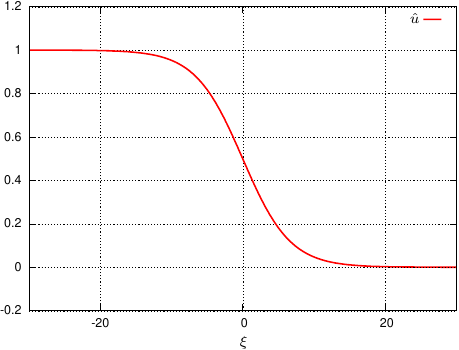}}
  \caption{%
    The solution $\uh(\xf)$ given by \eq{RXD_1} for $\g = 1$, $\h=0$,
    $\k=0.3$ and $\C = 0$. %
  }
  \figlabel{tanh}
\end{figure} 

Once $\uh(\xf)$ is known, we can find $\vh(\xf)$ using \eq{RXD_9}, as
\begin{equation}
\vh(\xf) = \vs + \frac{\Du}{\c} \uh'(\xf) . 
\eqlabel{RXD_17}
\end{equation}
Obviously, the profile of component $\vh$ represents not a wave front but a
pulse. In accordance with \eq{RXD_8}, we have
\[
  \vh \left( \pm \infty \right) = \vs .         
\]

\section{On the stability of the front solutions: continuous spectrum}
\seclabel{stability} 

The stability of any front solution we seek shall depend, in
particular, on the stability of its asymptotic spatially uniform steady states,
that is, on the continuous spectrium. This, unlike the
  discrete spectrum, is easily done analytically.  The system
\eq{RXD_2} can be written in the matrix form
\[
  \w_\t = \F(\w) + \D \w_{\x\x},               
\]
where
\begin{align*}
  \w= \Mx{ \u \\ \v }, \quad
  \F(\w) = \Mx{ \f(\u) - \v \\ 0}, \quad 
  \D = \Mx{0 & \Dv \\ -\Du & 0}.
\end{align*}
Suppose $\weq = [\ueq,\veq]^T$ is an equilibrium, i.e. $\F(\weq) = 0$.  We
perturb this point,
\[
  \w = \weq +  \wt ,
\]
and in the linear approximation we have
\begin{equation}
  \wt_\t = \F'(\weq) \wt + \D \wt_{\x\x} , \eqlabel{RXD_24}
\end{equation}
where $\F'=\Mx{\partial\F/\partial\w}$ is the Jacobian matrix.
By separation of variables, particular solutions of \eq{RXD_24} bounded in
space can be written as linear combinations of
\begin{equation}
  \wt(\x,\t) = \e^{\iu\wavenumber \x} \e^{\lam \t} \Mx{\C_1\\\C_2},  \qquad
  \wavenumber\in\Real, \quad \lam, \C_1, \C_2 \in\Complex . \eqlabel{RXD_25}
\end{equation}

\begin{figure}[tbhp]
  \centerline{\includegraphics{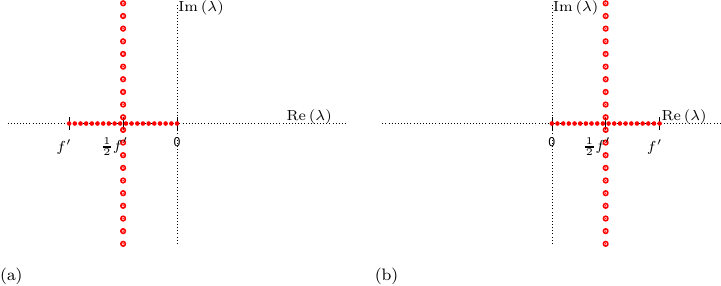}}
  \caption{%
    The continuous spectrium of an equilibrium, for %
    (a) $\f'=\f'(\ueq)<0$, %
    (b) $\f'=\f'(\ueq)>0$, %
    according to \eq{RXD_26}. %
  } \figlabel{spec}
\end{figure}

Substituting \eq{RXD_25} in \eq{RXD_24}, gives and eigenvalue problem
\begin{equation}
  \Mx{
    \f'(\ueq) & -1-\wavenumber^2 \Dv \\
    \wavenumber^2 \Du & 0
  }\Mx{ \C_1 \\ \C_2 } = \lam \Mx{\C_1 \\ \C_2} , 
\end{equation}
where $\f'=\partial\f/\partial\u$, and the eigenvalues are
\begin{equation}
  \lam_{1,2} = \tfrac{1}{2} \left[
  \f'(\ueq) \pm \sqrt{(\f'(\ueq))^2
    - 4\wavenumber^2 \Du 
    - 4\wavenumber^4 \Du \Dv
  } \right] ,
  \eqlabel{RXD_26}
\end{equation}
see \fig{spec}.  Therefore, if $\f'(\ueq)$ is positive, then
$\Re{\lam_{1,2}}\ge0$ and the steady state $(\ueq,\veq)$ is unstable, and if
$\f'(\ueq)$ is negative then $\Re{\lam_{1,2}}<0$ for all $\wavenumber\ne0$,
and the state is stable in linear approximation. Of course, even if both
asymptotic states are stable, the stability of the whole front solution will
still depend on the discrete spectrum; this is outside our scope.

\section{Fixing the polynomial reaction term}
\seclabel{Correspondent}

In this section we will find the particular form of the polynomial function
$\fh(\uh)$, as well as the parameters $\A$ and $\B$ that satisfy
\eq{RXD_13}.  To achieve this, we substitute \eq{RXD_15} into
\eq{RXD_13}, which gives
\begin{equation}
  \k( \uh - \g) ( \uh - \h ) \left\lbrace
    \A \left[
      \k^2 \left( 2 \uh - \g-\h\right)^2 + 2\k^2 ( \uh - \g ) ( \uh - \h)
    \right]  + \B
  \right\rbrace  = \fh(\uh). \eqlabel{RXD_19}
\end{equation}
We take our quartic polynomial $\fh(\uh)$ in the following form:
\begin{equation}
  \fh(\uh) = \sig (\uh - \ur1 )(\uh - \ur2 ) (\uh - \ur3 ) (\uh - \ur4 )  
  \eqlabel{RXD_20} 
\end{equation}
where $\left\{\ur1,\ur2\right\}=\left\{\g,\h\right\}$ and without loss of
generality $\sig = \pm1$; a different scaling of $\f$ would just result in a
change of the spatial and temporal scale of the solutions.

By substituting \eq{RXD_20} into \eq{RXD_19}, we obtain
\begin{equation}
\begin{aligned}
  & \k(\uh-\g)(\uh-\h) \left\lbrace
    \A \left[
      \k^2 \left(2\uh-\g-\h\right)^2 +2\k^2(\uh-\g)(\uh-\h)
    \right] + \B \right\rbrace
  \\ &
  = \sig (\uh-\ur1)(\uh-\ur2)(\uh-\ur3)(\uh-\ur4) .
\end{aligned}
\eqlabel{RXD_poly_11}
\end{equation}
By equating like terms we obtain
\begin{equation}
\begin{aligned} \mbox{}
[\uh^4] :\; & \frac{6\k^3 \Du \Dv}{\c} = - \sig ; \\[0.5em]
[\uh^3] :\; & \frac{12 \k^3 \Du \Dv (\g+\h)}{\c} = - \sig (\ur1 + \ur2 + \ur3 + \ur4) ; \\[0.5em]
[\uh^2] :\; & \frac{\k^3 \Du \Dv (7\g^2 + 22 \g\h + 7 \h^2)}{\c} + \frac{\k(\c^2 - \Du)}{\c} \\
            & = - \sig (\ur1\ur2+\ur1\ur3+\ur1\ur4+\ur2\ur3+\ur2\ur4+\ur3\ur4) ; \\[0.5em]
[\uh^1] :\; & \frac{\k^3 \Du \Dv (\g+\h)(\g^2 + 10 \g\h + \h^2)}{\c} - \frac{\k(\g+\h)(-\c^2+\Du)}{\c} \\
            & = - \sig (\ur1\ur2\ur3 + \ur1\ur2\ur4 + \ur1\ur3\ur4 + \ur2\ur3\ur4 ) ; \\[0.5em]
[\uh^0] :\; & \frac{\k^3 \g\h \Du \Dv(\g^2 + 4\g\h + \h^2)}{\c} + \frac{\g\k\h(\c^2-\Du)}{\c}
              = - \sig \ur1\ur2\ur3\ur4.
\end{aligned}
\eqlabel{RXD_21}
\end{equation}
This imposes five constraints onto a set of 11 parameters $\k$, $\g$, $\h$,
$\sig$, $\Du$, $\Dv$, $\ur1$, $\ur2$, $\ur3$, $\ur4$ and $\c$; hence we can
describe all solutions of this system by assigning six of these parameters as
free, and then finding the remaining five parameters as dependent on these six
free parameters.  We restrict consideration to real values of parameters in
both groups, except possibly the roots $\ur{3,4}$.  Moreover, as parameters
$\g$ and $\h$ fix the positions of the pre- and post-front resting states of
the solution \eq{RXD_16}, it convenient to have these two among the free
parameters; note also that we have already constrained $\sig$ to $\pm1$.

\section{Possible types of solutions}
\seclabel{Possibilities_of_Generalising}

As discussed in the Introduction, this study is not motivated by any
real-world applications leading to specific examples of reaction
cross-diffusion systems.  Rather, we are interested in theoretical
possibilities achievable within a certain class of models.  With this in mind,
we want to see if we can make the reaction cross-diffusion system with quartic
polynomial to look like generalizations, in one sense or another, of other
well-known models, from the much better studied class of systems with
self-diffusion. We shall say that we ``imitate'' those models.  The models
that we want to imitate are Fisher-KPP and ZFK-Nagumo.  Those models exhibit
propagating front solutions with asymptotics
\[
  \u (\xf \rightarrow + \infty)  = 0, \qquad  \u (\xf \rightarrow - \infty)  = 1.
\]
If we identify the scalar field $\u$ here with the namesake first dynamic
variable in our system, then this property can be achieved by letting $\g=0$
and $\h=1$ in \eq{RXD_16}.

We found in the previous section that the stability of a spatially uniform
steady state depends on the sign of the derivative of the quartic polynomial
at that state.  In terms of stability, to imitate the ZFK-Nagumo wave, we
would need a stable pre-front state and a stable post-front state, and
consequently an unstable equilibrium in between.  To imitate the Fisher-KPP
wave front we would need to have an unstable pre-front state and a stable
post-front state, with either no or two equilibria in between.  In this respect,
the possibilities for front waves from the reaction cross-diffusion system
with quartic polynomial are constraint by the following proposition.

\begin{proposition} \proplabel{proposition} %
  If the boundary-value problem \eq{RXD_10} with the nonlinearity defined
  by \eq{RXD_20} and \eq{RXD_poly_11} has a travelling wave front
  solution of the form \eq{RXD_16}, then the two asymptotic resting states
  $\{\g,\h\}$ are either the two outer roots of the quartic polynomial
  $\fh(\uh)$, or its two inner roots.
\end{proposition}

\begin{proof} %
  From \eq{RXD_poly_11}, among the roots of $\fh(\uh)$ we have
  $\left\{\ur1,\ur2\right\}=\left\{\g,\h\right\}$, and the other two roots,
  $\ur{3,4}$, are the roots of the quadratic in the square brackets, which is
  equivalent to
  \[
    \uh^2
    -(\g+\h)\uh
    +\frac{\g^2+4\g\h+\h^2+ \B/(\A\k^2)}{6} = 0.
  \]
  Hence $\frac12(\ur3+\ur4)=\frac12(\g+\h)$. If $\ur{3,4}\in\Real$,
  $\ur3\ne\ur4$, this implies that either $\g$ and $\h$ are two inner roots
  while $\ur3$ and $\ur4$ are the two outer roots, or vice versa. If
  $\ur3=\ur4$ the $\g$ and $\h$ are the two outer roots out of the three, and
  if $\ur{3,4}\not\in\Real$, then $\g$ and $\h$ the only two, therefore
  automatically the outer, roots.
\end{proof}

From \prop{proposition}, we conclude that of the resting states of the
front wave solution, only one can be stable but not both.  That means, in the
considered reaction cross-diffusion system with the quartic polynomial, it is
impossible to imitate ZFK-Nagumo wave in terms of the stability of the resting
states, but there is a chance to imitate Fisher-KPP wave. We note, however,
that for any given set of parameters of the model, the speed of the front
solution is in any case uniquely fixed by~\eq{RXD_33}, and this feature is
characteristic of ZFK-Nagumo fronts rather than Fisher-KPP fronts.

\section{Choice of Signs to Imitate Fisher-KPP}
\seclabel{Imitate_Fisher}
We have found that there is a possibility to imitate Fisher-KPP front wave, in
terms of the stability of the pre-front and post-front equilibria, by reaction
cross-diffusion system with quartic polynomial nonlinearity.  In this section,
we will turn this possibility into reality, by identifying appropriate
parameter choices.

Firstly, let us make sure that solution given by \eq{RXD_16} satisfies the
asymptotic boundary conditions of Fisher-KPP front wave,
\begin{equation}
  \uh(+\infty)=0, \qquad \uh(-\infty)=1. 
  \eqlabel{RXD_30}
\end{equation}
In \secn{Correspondent} we found that six parameters in \eq{RXD_21} can
be arbitrary assigned. We choose $\k, \g$ and $\h$ as three of such free
parameters, in order to satisfy \eq{RXD_30}. We have already committed
ourselves to the choice $\{\g,\h\}=\{0,1\}$, and we require $\k\ne0$.
\Tab{sec_05} lists the resulting four \textit{a priori} possibilities.

\begin{table}[htbp]
  \caption{%
    Examining possible choices to imitate
    Fisher-KPP front. The symbols $(\nearrow)$ and $(\searrow)$ mean
    that $\ch(\xf)$ is an increasing or decreasing function,
    respectively. %
  }
  \begin{center}
    \begin{tabular}{|c||c|c|c||c|c|c|}
      \hline
      \multicolumn{4}{|c||}{Choices}  & \multicolumn{3}{|c|}{Results}  \\ \hline
                   & $\g$ & $\h$ & $\k$ & $\ch$ & $\u(+\infty)$ & $\u(-\infty)$ \\ \hline
      $\mathrm{I}$ & $1$ & $0$ & $(+)$ & $ \nearrow $ & $0$ & $1$ \\ \hline
      $\mathrm{II}$ & $1$ & $0$ & $(-)$ & $ \searrow $ & $1$ & $0$ \\ \hline
      $\mathrm{III}$ & $0$ & $1$ & $(+)$ & $ \searrow $ & $0$ & $1$ \\ \hline
      $\mathrm{IV}$ & $0$ & $1$ & $(-)$ & $ \nearrow $ & $1$ & $0$  \\ \hline
    \end{tabular}
  \end{center}
  \tablabel{sec_05}
\end{table}

Clearly, choices that comply with \eq{RXD_30} are $(\mathrm{I})$
and $(\mathrm{III})$.  In both cases, equation \eq{RXD_15} gives
\begin{equation}
  \y(\uh) = \k \uh(\uh-1), \quad
  \y'(\uh)= 2\k(\uh-1), \quad
  \y''(\uh) = 2\k ,
  \qquad  \k>0.
  \eqlabel{RXD_31}
\end{equation}
The quartic polynomial $\fh(\uh)$ posited in \eq{RXD_20} allows $\sig=1$ or
$\sig=-1$.  Remember that the equation for the coefficients at $\uh^4$ in
\eq{RXD_21} states
\begin{equation}
  6 \k^3 \Du \Dv  = -\sig \c . \eqlabel{RXD_32}
\end{equation}
If $\sig = 1 $ then the solution \eq{RXD_16} will not satisfy the condition
\eq{RXD_30}: since $\Du ,\Dv$ and $\c$ are positive, equation
\eq{RXD_32} implies $\k < 0$, which is inconsistent with \eq{RXD_31}.

So, we must choose $\sig = -1$, which together with
$\{\g,\h\}=\{\ur1,\ur2\}=\{0,1\}$ turns the system \eq{RXD_21} to
\begin{align*}
  \frac{6\k^3\Du\Dv}{\c} &= 1 , \\
  \frac{12\k^3\Du\Dv}{\c} &= 1+\ur3+\ur4  ,    \\
  \frac{6\k^3\Du\Dv}{\c}+ \frac{\k^3\Du\Dv}{\c} -\k \frac{-\c^2+\Du}{\c} &= \ur3+\ur4+\ur3\ur4 , \\
  \frac{\k^3\Du\Dv}{\c} -\k \frac{-\c^2+\Du}{\c} &= \ur3\ur4 , \\
  0 &= 0 . 
\end{align*}
Previously, we let variables $\g$, $\h$ and $\k$ be free parameters.  We now
add to that list $\Du$ and $\Dv$. The rest of the variables will be dependent
on those as follows:
\begin{align}
  \c  &= 6 \k^3 \Du \Dv , \eqlabel{RXD_33} \\
  \ur3 &= \tfrac{1}{2} - \tfrac{1}{6}\sqrt{3 + 36 \rh} , \eqlabel{RXD_40} \\
  \ur4 &= \tfrac{1}{2} + \tfrac{1}{6}\sqrt{3 + 36 \rh} . \eqlabel{RXD_41} \\
\end{align}
where 
\begin{equation}
  \rh = \frac{\k(\Du -\c^2)}{\c} . \eqlabel{RXD_38}
\end{equation} 

The quartic polynomial now has the form
\begin{equation}
\fh(\uh) = -\uh(\uh-1)(\uh-\ur3)(\uh-\ur4) ,
\eqlabel{RXD_39}
\end{equation}
where $\ur3$ and $\ur4$ are given by \eq{RXD_40} and \eq{RXD_41}.

We expect that, in principle, if the quartic polynomial is substituted into
the system \eq{RXD_2} , i.e.
\begin{equation}
\begin{aligned}
  \u_\t &=  - \u(\u-1)(\u-\ur3)(\u-\ur4) + \vs  -\v  + \Dv \v_{\x\x} \: ,\\
  \v_\t &= -\Du \u_{\x\x} \: ,
\end{aligned}
\eqlabel{RXD_42}
\end{equation}
then the solution of \eq{RXD_42} is a front wave which imitates the
front wave in Fisher-KPP with respects to the 
stability of the pre-front and post-front resting states.

The choices of values of the given parameters change the values of the roots
$\ur3$ and $\ur4$, which leads to one of the following cases.
\begin{enumerate}[leftmargin=5em,label=Case \Roman*:,ref=\Roman*]
\item\caselabel{inner} If $\rh \in ( \frac{1}{6} ,  + \infty)$, then
  $\ur{3,4}\in\Real\setminus[0,1]$ and the restings states $\{0,1\}$ are inner
  roots. 
\item\caselabel{double-double} If $\rh=\frac{1}{6}$, then
  $\{\ur3,\ur4\}=\{0,1\}$ and the resting states $\{0,1\}$ are the only two
  roots, both double.
\item\caselabel{outer-of-four} If $(\rh \in (-\tfrac{1}{12},\tfrac{1}{6})$, then 
  $\ur{3,4}\in(0,1)$, $\ur3\ne\ur4$, and the resting states  $\{0,1\}$ are
  outer of four roots. 
\item\caselabel{outer-of-three} If $\rh=-\tfrac{1}{12}$, then 
  $\ur3=\ur4=\tfrac12$, and the resting states  $\{0,1\}$ are
  outer of three roots. 
\item\caselabel{outer-of-two} If $\rh \in (-\infty, -\tfrac{1}{12})$, then 
  $\ur{3,4}\in\Complex\setminus\Real$ and the resting states  $\{0,1\}$ are
  the only two roots. 
\end{enumerate}
Remember that by virtue of \eq{RXD_38} and \eq{RXD_33}, this means that
the location of the roots $\ur{3,4}$ is determined by the three parameters
$\k , \Du$ and $\Dv$.

\section{Stability of the Resting States}
\seclabel{Stability_of_zero_one}

Previously, we have linearised the system \eq{RXD_2} for general function
$\f(\u)$ about an equilibrium and derived the formula of the eigenvalues
\eq{RXD_26}.  Substituting the quartic polynomial function \eq{RXD_39}
into the function of the eigenvalue yields that, the eigenvalues of the
equilibrium $\ur1 = 0$ are given by
\begin{equation}
  \lam_{1,2} = \frac12 \left[ \ur3\ur4 \pm \sqrt{\ur3^2\ur4^2
    - 4\wavenumber^2\Du
    - 4\wavenumber^4\Du\Dv
  }\right],
\eqlabel{RXD_46}
\end{equation}
while the eigenvalues of the equilibrium $\ur2 = 1$ are given by 
\begin{equation}
  \lam_{1,2} = \frac12 \left[ -(1-\ur3)(1-\ur4) \pm \sqrt{(1-\ur3)^2(1-\ur4)^2
    - 4\wavenumber^2\Du
    - 4 \wavenumber^4 \Du \Dv
  } \right].
\eqlabel{RXD_47}
\end{equation}

In the ``inner roots'' \case{inner}, the two roots $\ur3$ and $\ur4$
have different signs, and are to opposite sides of $1$.  Thus, from
\eq{RXD_46} and \eq{RXD_47} we deduce that the pre-front $\ur1=0$ is
stable and the post-front $\ur2=1$ is unstable.

The similarity between Fisher-KPP and inner roots case is that both systems
have two consecutive roots of $\fh(\u)$ that coincide with the resting states
of a wave front.  The difference between them is that the pre-front in
Fisher-KPP is unstable and the post-front is stable, while in inner
roots case it is the other way round,
the pre-front is stable and the post-front is unstable.

In the ``outer roots'' cases~\caseref{outer-of-four} and
\caseref{outer-of-three} as well as ``the only two roots''
\case{outer-of-two}, wee see from \eq{RXD_46} and \eq{RXD_47}
that the pre-front $\ur1=0$ is unstable and the post-front $\ur2=1$ is
stable. This matches the stability of the equilibria in Fisher-KPP model.

The marginal \case{double-double} gives $\Re{\lam_{1,2}}=0$ so the
stability of the resting states cannot be established in linear approximation,
and requires separate consideration. We leave this outside the scope of this
paper.

\Tab{stability_quartic_wave} sums up the results of above
analysis.

\begin{table}[htbp]
  \caption{%
    The stability of the resting states in the front wave depends on the choice
    of the roots of the quartic polynomial. %
  }
  \begin{center}
    \begin{tabular}{|c|c|c|c|}
      \hline 
      {\bf Choice of roots} & {\bf Pre-front} & {\bf Post-front} & {\bf Matching with Fisher-KPP}
      \\[0.7em] \hline
      \Case{inner}: Inner & stable & unstable & \xmark
      \\[0.4em] \hline
      \Case{outer-of-four}: Outer & unstable & stable &  \cmark 
      \\[0.4em] \hline
      \Case{outer-of-three}: Double & unstable & stable &  \cmark 
      \\[0.4em] \hline
      \Case{outer-of-two}: Complex & unstable & stable &  \cmark 
      \\ 
      \hline
    \end{tabular}
  \end{center}
  \tablabel{stability_quartic_wave}
\end{table}

In the next section we will show the result of the numerical
simulation for each case.

\section{Numerical Simulations}
\seclabel{Num_Simul}

\subsection{General settings}

We simulate numerically the reaction cross-diffusion system
\begin{equation}
\begin{aligned}
  \u_\t &= \f(\u,\v) + \Dv \v_{\x\x} , \\ 
  \v_\t &= -\Du \u_{\x\x} ,
\end{aligned}
\eqlabel{RXD_48}
\end{equation}
for $- \a \leq \x \leq \b$ and $\t \geq 0$, where the kinetic term $\f(\u,\v)$
is quartic polynomial
\[
  \f(\u,\v) = - \u(\u-1)(\u-\ur3)(\u-\ur4)  - \v ,
\]
and $\ur3$ and $\ur4$ are dependent parameters defined in
\eq{RXD_40} and \eq{RXD_41}.
We apply no-flux boundary conditions,
\[
  \u_\x(-\a,\t) = \u_\x(\b,\t) = \v_\x(-\a,\t) = \v_\x(\b,\t) = 0,
\]
and the initial condition taken from the analytical solution, that is 
\[
  \u(\x,0) = \uh(\x),
  \qquad
  \v(\x,0) = \vh(\x),
\]
where $\uh$ and $\vh$ are defined in \eq{RXD_16} and \eq{RXD_17}.

\begin{table}[htbp]
  \caption{%
    Parameters and equilibria in numerical simulations. %
  }
  \begin{center}
  \begin{tabular}{|c|c|c|c|c|}\hline
    Case &
           \caseref{inner} &
           \caseref{outer-of-four} &
           \caseref{outer-of-three} &
           \caseref{outer-of-two} \\\hline
    Figure(s) & \figref{Inner_roots_TW},%
                \figref{proof_dynamical_instability_COMPLEX},%
                \figref{proof_dynamical_instability_INNER} &
           \figref{proof_dynamical_instability_outer} &
           \figref{Double_roots_TW},%
                \figref{proof_dynamical_instability_double} &
           \figref{Complex_roots_TW} \\\hline
    $\k$ & $1$ & $1$ & $1$ & $1$\\\hline
 $\Du$ & $1.25$ & $0.2$ & $2.917$ & $0.4$\\\hline
 $\Dv$ & $0.1$ & $0.35$ & $0.1$ & $1.5$\\\hline
 $\ur1$ & $0$ & $0$ & $0$ & $0$\\\hline
 $\fh'(\ur1)$ & $-0.75$ & $0.11$ & $0.25$ & $3.656$\\\hline
 $\ur2$ & $1$ & $1$ & $1$ & $1$\\\hline
 $\fh'(\ur2)$ & $0.75$ & $-0.11$ & $-0.25$ & $-3.656$\\\hline
 $\ur3$ & $1.5$ & $0.874$ & $0.5$ & $0.5+1.845\iu$\\\hline
 $\fh'(\ur3)$ & $-1.5$ & $0.083$ & $0$ & \\\hline
 $\ur4$ & $-0.5$ & $0.126$ & $0.5$ & $0.5-1.845\iu$\\\hline
 $\fh'(\ur4)$ & $1.5$ & $-0.083$ & $0$ & \\\hline
\end{tabular}
  \end{center}
  \tablabel{par-eq}
\end{table}

We will show the results of the simulation for cases
\caseref{inner}, %
\caseref{outer-of-four}, %
\caseref{outer-of-three} and %
\caseref{outer-of-two} %
identified above.  For each case, we pick an appropriate set of values of the
free parameters to satisfy the correspoinding
conditions. \Tab{par-eq} lists the parameter values used and the
corresponding equilibria. Note that the value of $\Du$ for Case
\caseref{outer-of-three} in the table is given to three decimal places; in
fact it was determined from the exact condition that $\rh=-1/12$, which
implies
\begin{equation}
\eqlabel{cond_for_double_root}
\Du = \frac{2 + \k^2 \Dv}{72 \k^6\Dv^2}.
\end{equation}

The numerical simulations are done using finite differences, fully explicit
first order for time and second order central for space. The space
discretization interval is $[-\a,\b]=[-37.5,150]$ and the discretisation steps
are $\Delta \x = 0.15$ and $\Delta \t = 4 \times 10^{-6}$ unless otherwise
stated.  The choice of the discretization steps is motivated by the numerical
stability and accuracy analysis of the scheme, which will be presented later.

\subsection{The inner roots case}
\seclabel{Inner}

\begin{figure}[tbhp]
  \centerline{\includegraphics{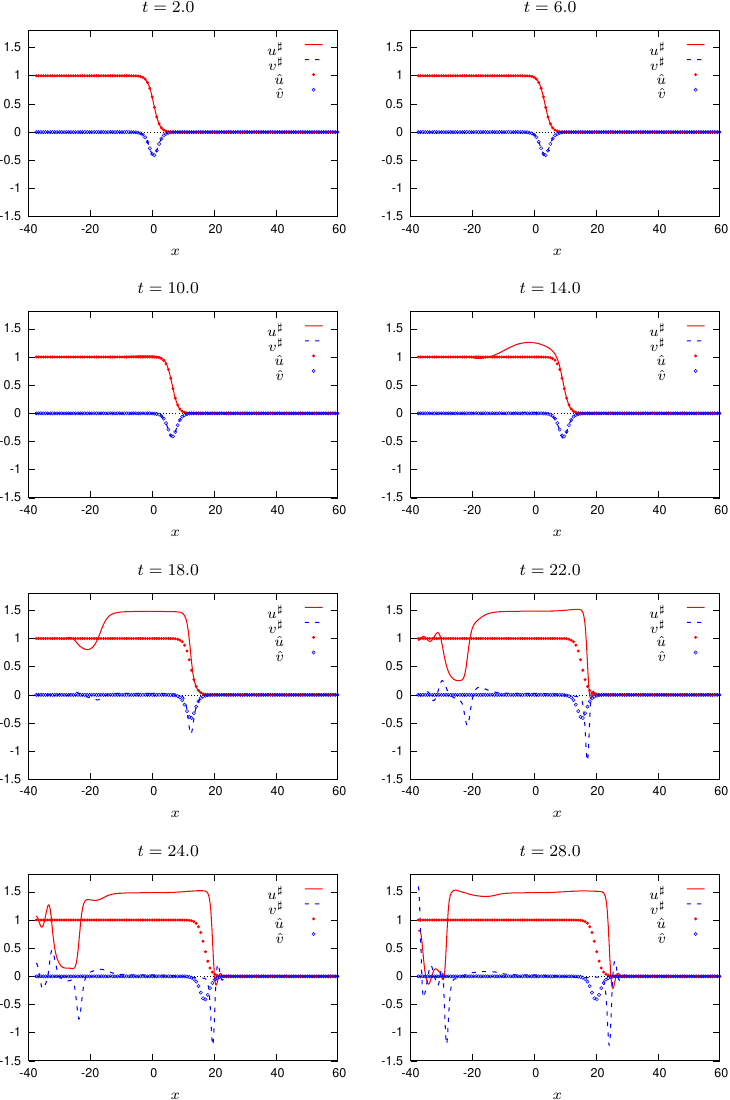}}
  \caption{%
    The numerical simulation of reaction cross-diffusion system with quartic
    polynomial where the resting states of the front coincides with the {\bf
      inner roots} of the quartic. The values of parameters in this
    simulations are $\Du=1.25$, $\Dv =0.1$ and $\k=1$.  Here and in the
    subsequent figures, $\unum=\unum(\x,\t)$, $\vnum=\vnum(\x,\t)$ is the
    numerical solution, whereas $\uana=\uana(\x-\c\t)$, $\vana=\vana(\x-\c\t)$
    is the analytical solution used as the initial condition for the
    numerics.  %
  }
  \figlabel{Inner_roots_TW}
\end{figure}

As shown above, in this case the pre-front equilibrium $\ur1=0$ is stable,
while the post-front equilibrium $\ur2=1$ is unstable. Hence we expect in
simulations that the post-front state evolves to another, stable equilibrium.
This is indeed what happens in simulations, see \fig{Inner_roots_TW}.

For the parameters used in this simulations, the unstable equilibrium $\ur2=1$
is surrounded by the pre-front equilibrium $\ur1=0$ and the upper stable
equilibrium $\ur4=1.5$.  Thus in this case we expect the post-front state
attracted to either of these two stable equilibria.

In fact, the solution curiously does both, i.e. is first attracted to the upper
stable equilibrium, $\ur4=1.5$, but does not stay there for long and departs
for the lower stable equilibrium, $\ur1=0$. As a result, a pulse-shaped
solution develops, with the pre-front and post-front states at $\ur1=0$, and
the plateau state near $\ur4=1.5$. This phenomenology is similar to that
observed in \cite{Biktashev-Tsyganov-2005} for excitable (i.e. one stable
equilibrium) cross-diffusion systems, incluiding oscillatory front and
oscillatory back, both trigger waves from one stable equilibrium to another
--- and is of course very far from the initial condition which is a monotonic
front from a stable equilibrium to an unstable one.

\subsection{The Result of Simulation of Distinct Real Roots, Double Roots and Complex Roots}
\seclabel{rest_cases}

\begin{figure}[tbhp]
  \centerline{\includegraphics{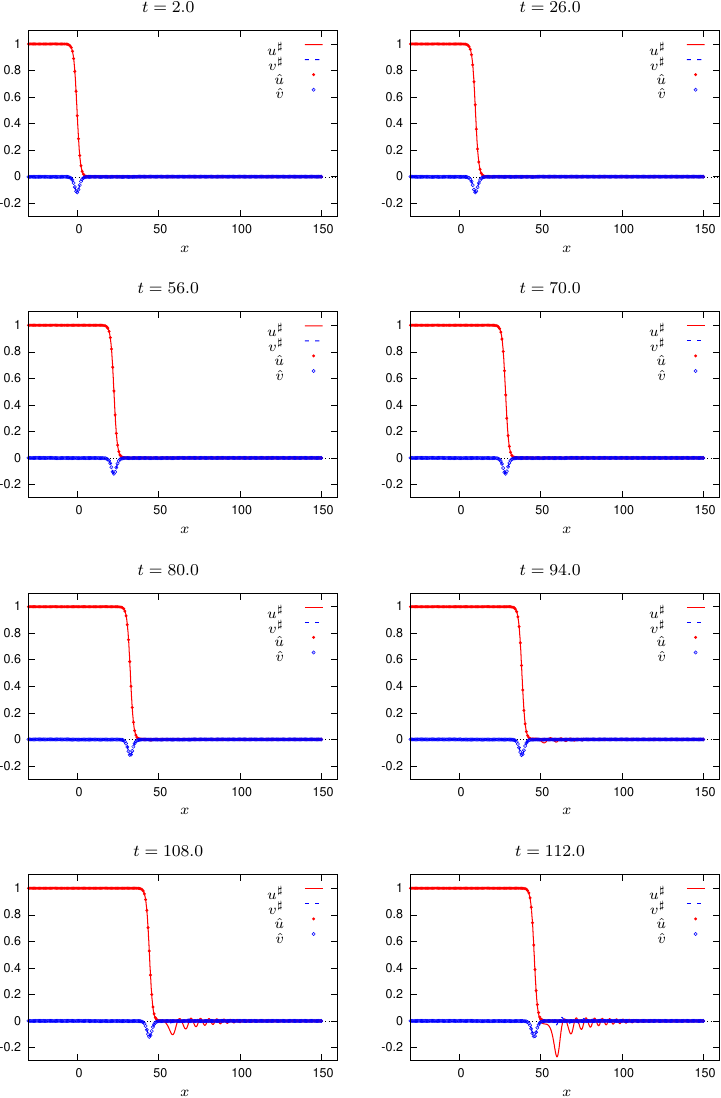}}
  \caption{%
    The numerical simulation of reaction cross-diffusion system with quartic
    polynomial where the resting states of the front coincides with the {\bf
      outer roots} of the quartic. The values of parameters in this
    simulations are $\Du=0.2$, $\Dv =0.35$ and $\k=1$. %
  }
  \figlabel{outer_roots_TW}
\end{figure}

The behaviour of the propagating wave front for the distinct real roots case
and double roots case is quite similar.  The simulation shows that the
numerical propagating wave remains close to the analytical wave for a period
of time. Then an oscillation appears near the onset of the front. After that the
oscillation grows as the time evolves, which causes the numerical solution to
break up.  The results of the simulation of distinct real roots case is shown
in \fig{outer_roots_TW} while the results of double roots case is shown
in \fig{Double_roots_TW}.

For complex roots case, we observe that the instability occurs earlier than
all previous cases (inner roots case, outer roots case and double roots case).
Moreover, the numerical front does not last as long as those front waves in
the other cases, see \fig{Complex_roots_TW}.

\begin{figure}[tbhp]
  \centerline{\includegraphics{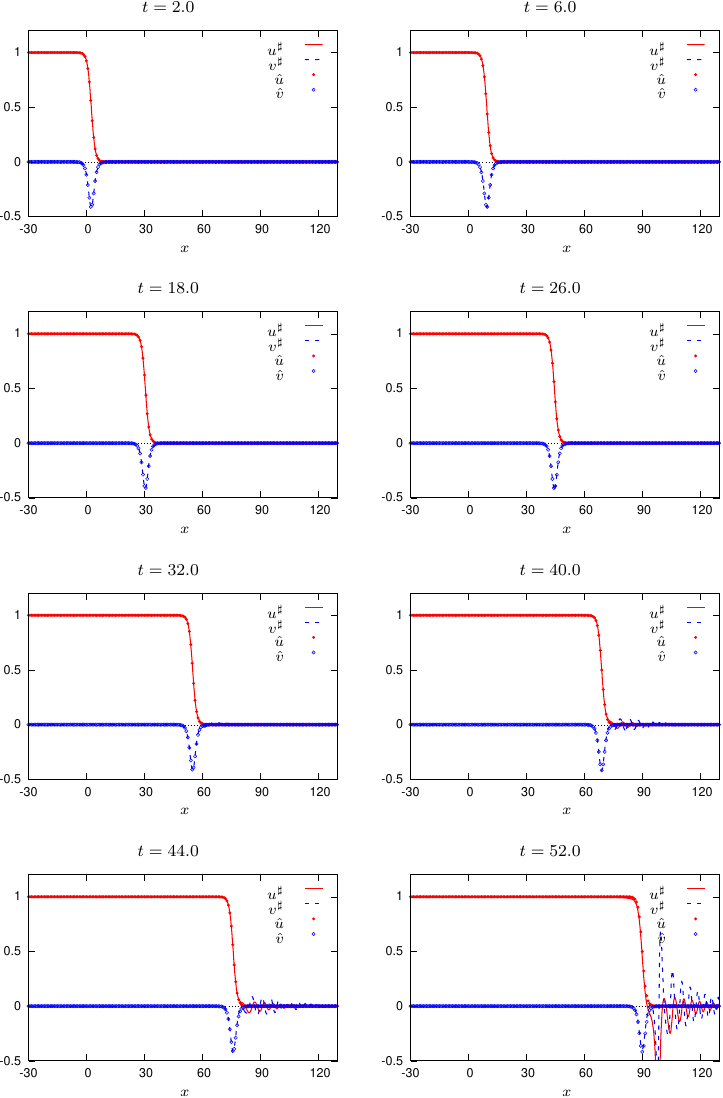}}
  \caption{%
    The numerical simulation of reaction cross-diffusion system with quartic
    polynomial where there are {\bf double roots} and the resting states are
    simple roots. The values of parameters in these simulations are $\Dv =0.1$
    and $\k=1$ where $\Du$ is given in the formula
    \eq{cond_for_double_root}. %
  }
  \figlabel{Double_roots_TW}
\end{figure}

\begin{figure}[tbhp]
  \centerline{\includegraphics{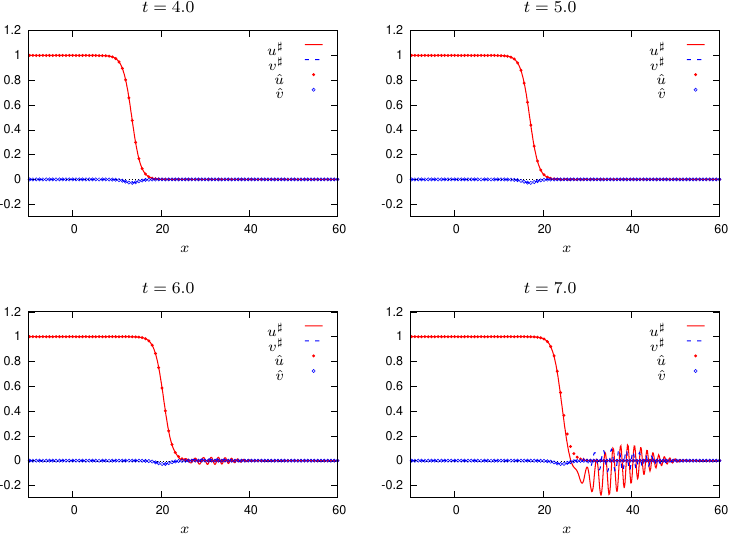}}
  \caption{%
    The numerical simulation of reaction cross-diffusion system with quartic
    polynomial where there are two {\bf complex conjugate roots}. The values
    of parameters in these simulations are $\Du=0.4 , \Dv =1.5$ and
    $\k=1$. The instability make the numerical solution run away at $\t= 8$.
  }
  \figlabel{Complex_roots_TW}
\end{figure}  

\section{The instability of the solution}
\seclabel{instability_FRONT_outer}

In the previous sections we have shown the results of direct numerical
simulation on reaction cross-diffusion system \eq{RXD_48} where the initial
condition is an exact analytical wave solution.  This analytical solution
presents a monotonic wave front.

We have considered four cases, corresponding to different positions of the
roots of the quartic polynomial.  In all four cases considered, there are
oscillations which appear near the onset of the wave front.  These oscillations
grow as time evolves, which obviously means
that the propagating wave front is not
stable.  We now would like to address the question whether this was due to
dynamical instability in the underlying partial differential equations, or
numerical instability, i.e. artefact of the numerical scheme used.

Our plan on how to distinguish numerical instability from the numerical is as
follows.  If the instability is numerical, then its features shall
significantly depend on details of the numerical scheme. For instance, the
oscillations could be reduced by changing the discretisation steps.
Conversely, the dynamical instability the behaviour of the solution may be
affected by refining the discretisation steps only slightly, if the simulation
is ``resolved''.

A crude theoretical analysis of numerical stability of the scheme we use can
be achieved by removing the kinetic terms from system \eq{RXD_48}. In this
way, we obtain the following
\begin{align*}
  \u_\t &= \Dv \v_{\x\x}, \\
  \v_\t &= -\Du \u_{\x\x}.
\end{align*}
For the forward-time, central-space discretization on the grid
$\x\in\stepx\,\Zahlen$, $\t\in\stept\,\Zahlen$, using the standard von Neumann
stability analysis, for the Fourier component $(\u,\v)\propto\e^{\iu\q\x}$ we
find the amplification factor $\ampf$, such that
\begin{equation} \eqlabel{RXD_51}
  \abs{\ampf(\q)}^2 = 1 + 16\Du\Dv\stept^2 \stepx^{-4} \, \sin^4\left(\q\stepx/2\right) ,
\end{equation}
which means that the numerical scheme is unstable as the condition
$\abs{\ampf} \leq 1$ will not be satisfied, in principle, for any choice of
discretization steps.
 
However, let us look at the quantitative aspect of the numerical instability.
Namely, let us estimate the time it takes for the numerical instability to
grow to macroscopic value. Supposing, for a crude estimate, that the seed of
the instability comes from round-off errors, so is of the order of machine
epsilon $\meps$, and it will become significant when it grows to an order of
$1$. Then, with the amplification factor $\ampf(\q)$, the number of time steps
required for that will be at least
$\ln\abs{1/\meps}/\max_{\q}\left(\ln\abs{\ampf(\q)}\right)$.
Taking the leading order approximation for the $\ln\abs{\ampf(\q)}$ in
\eq{RXD_51}, we get the time interval required for the instability to grow
to macroscopic size as
\[
  \Tinst \approx \frac{\ln\abs{\meps^{-1}} \, \stepx^4 }{8 \stept \Du \Dv } .
\]
By substituting the values of parameters we used in our simulation, we see
that in all cases $\Tinst$ is much bigger than the time $\Tbreak$ taken for
the numerical waves to break up.  \Tab{Comparison_T_inst_growth}
clarifies more by numbers.  We took $\meps=10^{-15}$.

\begin{table}[htbp]
  \caption{%
    Comparison between the theoretical instability time $\Tinst$, and time
    $\Tbreak$ to break-up in numerics, in the four selected simulations. %
  }
  \begin{center}
  \begin{tabular}{|c|c|c|} \hline 
    Case  & $\Tinst$ & $\Tbreak$ \\ \hline
Inner roots & $4371.3$ & $30$ \\\hline
Outer roots & $7805.9$ & $112$ \\\hline
Double roots & $1873.2$ & $52$ \\\hline
Complex roots & $910.7$ & $7$ \\\hline
\end{tabular}
  \end{center}
  \tablabel{Comparison_T_inst_growth}
\end{table}

This comparison suggests that even though the numerical scheme is
formally unstable, this instability cannot affect the
  numerical solutions on the time intervals involved. This means that
  there is no need to look for more sophisticated, stable methods to
  simulate the solutions presented. This also means that the
  numerical instability cannot explain the behaviour observed in our
numerics, and we must consider the possibility of a dynamical
instability.

So, according to our plan, we have verified the plausibility of a dynamical
instability by repeating the simultions at different discretization steps.  We
have repeated each of the simulations, once with bigger discretization steps
and once with smaller discretization steps. We have found that the behaviour
of the solution does not significantly change even after we refine the
discretisation.  More precisely, once the oscillations appear, we have found
the growth rate of the oscillation is the same in all different discretisation
steps. \Fig{proof_dynamical_instability_outer} illustrates that for
the ``outer roots'' case: even though the moment of onset of the instability
depends on the discretization, its growth rate is not affected by it.

\begin{figure}[tbhp]
  \centerline{\includegraphics{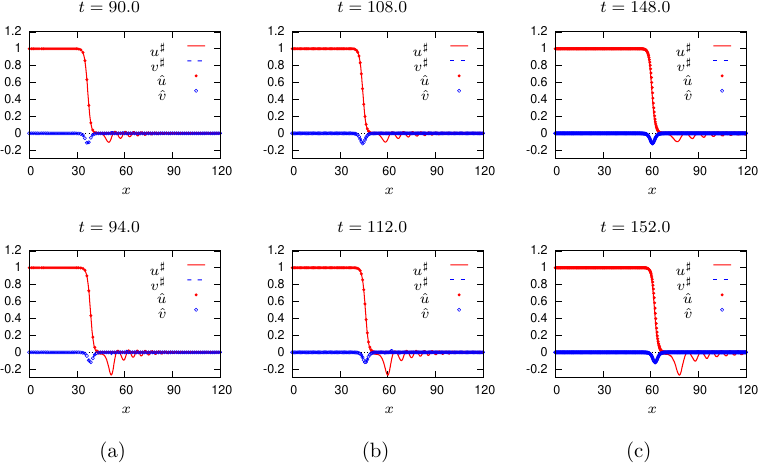}}
  \caption{%
    The dynamical instability appears for {\bf outer roots case}. The
    behaviour of the solution does not change even after the steps are
    refined. The values of parameters are $\k=1$ , $\Du = 0.2$ and
    $\Dv = 0.35$. The discretisation is: %
    (a) $\stepx = 0.25$, $\stept = 4 \times 10^{-5}$; %
    (b) $\stepx = 0.15$, $\stept = 4 \times 10^{-6}$; %
    (c) $\stepx = 0.05$, $\stept = 1 \times 10^{-7}$. %
  }
  \figlabel{proof_dynamical_instability_outer}
\end{figure}

The same thing happened in double roots case and complex roots case. Change of
discretisation steps changes the time of the onset of the instability, but not
the growth rate of the instability, as can be seen in 
\fig{proof_dynamical_instability_double} and \fig{proof_dynamical_instability_COMPLEX}.

\begin{figure}[tbhp]
  \centerline{\includegraphics{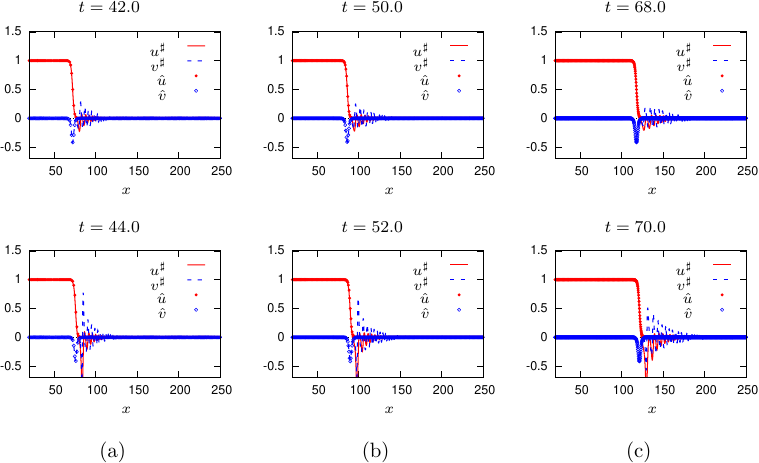}}
  \caption{%
    The dynamical instability appears for {\bf double roots case}. Each
    column represents the front wave for different discretisation
    steps. The behaviour of the solution does not change even if the
    steps are refined. The values of parameters are $\k=1$ and
    $\Dv = 0.1$. The discretisation is: %
    (a) $\stepx = 0.25$, $\stept = 4 \times 10^{-5}$; %
    (b) $\stepx = 0.15$, $\stept = 4 \times 10^{-6}$; %
    (c) $\stepx = 0.05$, $\stept = 1 \times 10^{-7}$. %
  }
  \figlabel{proof_dynamical_instability_double}
\end{figure}

\begin{figure}[tbhp]
  \centerline{\includegraphics{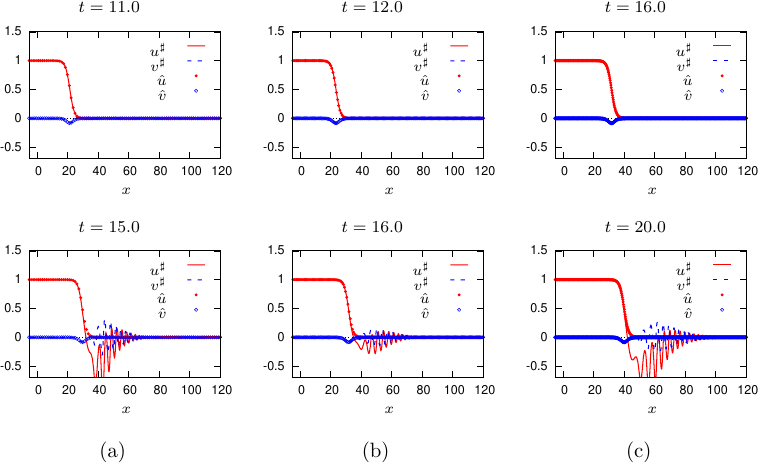}}
  \caption{%
    The dynamical instability appears for {\bf complex roots case}. Each
    column represents the front wave for different discretisation
    steps. The behaviour of the solution does not change even if the
    steps are refined. The values of parameters are
    $\k=1$, $\Du = 1.25$ and $\Dv = 0.1$. The discretisation is: %
    (a) $\stepx = 0.25$, $\stept = 4 \times 10^{-5}$; %
    (b) $\stepx = 0.15$, $\stept = 4 \times 10^{-6}$; %
    (c) $\stepx = 0.05$, $\stept = 1 \times 10^{-7}$. %
  }
  \figlabel{proof_dynamical_instability_COMPLEX}
\end{figure}

For the ``inner roots'' case, the initial condition is a front of invasion of
an unstable state into a stable state, and the numerical simulation show
behaviour different from other cases: now the instability appears, at first,
as the elevation of the $\u$-field right behind the front. So we observe how
this instability changes with different discretization steps.  The result is
shown in \fig{proof_dynamical_instability_INNER}. We see, again,
that the time of the onset of the instability does depend on the
discretization steps, but the growth rate remains the same. The subsequent
behaviour of the solution also remains qualitatively similar, involving
formation of a propagating pulse with a plateau and a back --- even though
shifted in time and differing in detail, which is of course only expectable
for a solution affected by a dynamical instability.

\begin{figure}[tbhp]
  \centerline{\includegraphics{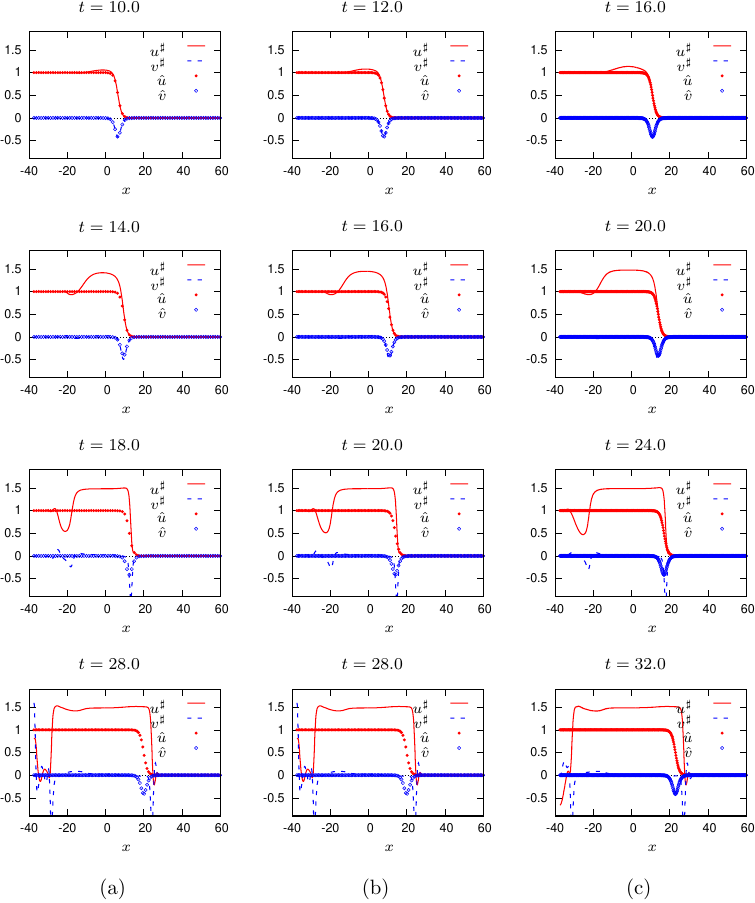}}
  \caption{%
    The dynamical instability appears for {\bf inner roots case}. Each column
    represents the front wave for different discretisation steps. The
    behaviour of the solution does not change even the steps are refined. The
    values of parameters are $\k=1$, $\Du = 1.25$ and $\Dv = 0.1$. The
    discretisation is: %
    (a) $\stepx = 0.25$, $\stept = 4 \times 10^{-5}$; %
    (b) $\stepx = 0.15$, $\stept = 4 \times 10^{-6}$; %
    (c) $\stepx = 0.05$, $\stept = 1 \times 10^{-7}$. %
  }
  \figlabel{proof_dynamical_instability_INNER}
\end{figure}

We can conclude that insofar as it may be established by numerical
simulations, the analytical front solutions are dynamically unstable:
they yield to solutions with oscillatory fronts,
  which are beyond the main scope of the current paper and requires
  separate study.

\section{Discussion}
\seclabel{discussion}

The main purpose of the paper, which has been successfully achieved, was to
demonstrate the feasibility, and provide an example, of constructing a PDE
model of a certain class which has desirable analytical solutions. Regardless
of the utility of the particular example we have considered, we hope that the
technique we used may be helpful in other problems similarly formulated.

More specifically, our aim has been a reaction-cross-diffusion system with a
polynomial nonlinearity, which would have solutions in the form of a
propagating front. We have found that to achieve that, the nonlinearity must
be at least quartic, in which case the system may indeed have solutions in the
form of  monotonic propagating fronts.  The situation is similar to
ZFK-Nagumo model rather than Fisher-KPP model in that for given parameters of
the system, the speed and shape of the front solution are uniquely defined.

We have further established that in terms of stability of pre-front and
post-front equilibria, the proposed model may be likened to the Fisher-KPP
system (one of the equilibria is stable and the other unstable) but not
ZFK-Nagumo (with both equilibria stable).

The quartic nonlinearity can be of various diffierent classes depending on
behaviour of its four roots: when the asymptotic equilibria are two inner
roots, two outer roots out of four, two outer roots out of three, the only two
simple roots (with the other two being complex) and two double roots.

We have made simulations of selected examples of the proposed model belonging
to different algebraic classes, and in all of these examples it happened that
the analytical solutions are dynamically unstable, with some of the
instabilities distinct from those related to the unstable pre-front
equilibrium. Since the conclusion about instability of the solutions is based
only on direct numerical simulations of arbitrarily selected examples, it
requires further investigation,  both theoretically and
numerically, perhaps including continuation of propagating wave
  solutions rather than  just direct numerical simulations,
and wider parametric searches.  A good survey of the relevant theory can be
found in~\cite{Sandstede-2002}, and examples of numerical tools
suitable for this task are AUTO~\cite{Doedel-etal-2007} and WAVETRAIN
\cite{Sherratt-2012}.

Returning to feasibility of proposed PDE system as a model of real processes,
we recall that KPP-Fisher is a viable model despite the unstable pre-front
state. As it is well known, there are two inter-related reasons for that. One
reason is the positivity of the equation: non-negative initial conditions
guarantee that the solution will remain non-negative at all times. Since the
linearly unstable pre-front state is $0$, i.e. at the border of the domain
invariant under the system, this motivates restriction on the class of
perturbations considered to those that would respect the positivity. The other
reason is also related to the fact that the pre-front state is $0$, but is of
physical rather than mathematical nature: it motivates applications in which
the dynamic field represent an essentially non-negative quantity with the
meaning of a concentration of some kind; specifically, in the seminal
papers~\cite{Fisher-1937,Kolmogorov-etal-1937} it was population
density. With that physical sense of the dynamic field, the magnitude of
physically feasible perturbations related to fluctuations must decay as the
system gets closer to the pre-front state, and exactly vanish at that
state. This motivates consideration of solutions in specially constructed
functional spaces that take this issue into account, in which the solution may
be stable --- despite the formal instability of the pre-front state in the
sense of generic dynamical systems theory. In this context, the possibility
of, and, as numerics show, preference for, the non-monotonic fronts is only
possible because the class of model we consider does not possess the
positivity property. Here we note that the models with linear cross-diffusion
cannot have that property in principle, see e.g.~\cite{Gorban-etal-2011}.

The above consideration motivates possible continuation of the present work:
\begin{itemize}
\item ZFK-Nagumo type fronts, i.e. monotonic fronts with stable pre-front and
  stable post-front states, may be sought for in models with polynomial
  nonlinearity of degrees higher than four;
\item Reasonably stable monotonic fronts switching from a zero pre-front state
  may be observed in models with nonlinear cross-diffusion,
  e.g. ``pursuit-evasion'' type mutual taxis of the components;
\item As the fronts actually observed in numerical simulations of
  cross-diffusion models so far are typically oscillatory, search of exact
  solutions of that kind would involve ``inventing'' an ansatz more
  sophisticated than that given by \eq{RXD_15} and \eq{RXD_16}.
\end{itemize}
All that should be considered in the context that the problem addressed in
this paper is about the ``fast subsystem'' in \eq{RXD_main}, and encompasses
just the first step in the singular perturbation theory in the limit
$\eps\to0$.

\section*{Acknowledgments}
A.A. is grateful to %
Prince Sattam Bin Abdulaziz University %
for sponsoring his Ph.D. studentship. %
V.N.B.'s work was supported in part by %
the EPSRC Grant No. EP/N014391/1 (UK), 
and National Science Foundation Grant No. NSF PHY-1748958, NIH Grant
No. R25GM067110, and the Gordon and Betty Moore Foundation Grant
No. 2919.01 (USA). 

\section*{References}


\end{document}